\documentclass[pra]{revtex4}

\usepackage{amsmath, amssymb, amsthm, amsfonts, latexsym, graphicx}
\newtheorem*{lem2}{Lemma 2}

\newcommand{\goes}{\:\longrightarrow\:}		% --> 
\newcommand{\half}{\frac{1}{2}}
\newtheorem{lemma}{Lemma}

\newcommand{\cH}{\mathcal{H}}

\newcommand{\ket}[1]{\left| #1\right\rangle}      % ket vector
\newcommand{\bra}[1]{\left\langle #1\right|}      % bra vector
\newcommand{\kets}[1]{| #1 \rangle}    				    % small ket vector
\newcommand{\bras}[1]{\langle #1 |}        				% small bra vector
\newcommand{\ii}{\mathbb{I}}											% the I with two vertical lines
\newcommand{\norm}[1]{\left\| #1\right\|}        	% norm
\newcommand{\ep}{\epsilon}        								% epsilon

\newcommand{\BQP}{{\textsf{BQP}}}
\newcommand{\QMA}{{\textsf{QMA}}}

\begin{document}

\title{Fast Universal Quantum Computation with Railroad-switch Local Hamiltonians}

\author{Daniel Nagaj}
\affiliation{Research Center for Quantum Information, Institute of Physics, Slovak Academy of Sciences, D\'{u}bravsk\'{a} cesta 9, 845 11 Bratislava, Slovakia}
%\affiliation{Quniverse, L\'{i}\v{s}\v{c}ie \'{u}dolie 116, 841 04, Bratislava, Slovakia}
\email{daniel.nagaj@savba.sk}

\begin{abstract}
We present two universal models of quantum computation with a time-independent, 
frustration-free Hamiltonian. 
The first construction uses 3-local (qubit) projectors, and the second one requires 
only 2-local qubit-qutrit projectors.
We build on Feynman's Hamiltonian computer idea \cite{Feynman:85} and use a 
{\em railroad-switch} type clock register.
%, inspired by the triangle Hamiltonian of Eldar et al. \cite{ERtriangle}.
The resources required to simulate a quantum circuit with $L$ gates in this model are $O(L)$ small-dimensional quantum systems (qubits or qutrits), 
a time-independent Hamiltonian composed of $O(L)$ local, constant norm, projector terms, 
the possibility to prepare computational basis product states, a running time $O(L \log^2 L)$, and the possibility to measure a few qubits in the computational basis. 
Our models also give a simplified proof of the universality of 3-local Adiabatic Quantum Computation. 
\end{abstract}

\date{\today}

\maketitle

%%%%%%%%%%%%%%%%%%%%%%%%%%%%%%%%%%%%%%%%%%%%%%%%%%%%%%%%%
%%%%%%%%%%%%%%%%%%%%%%%%%%%%%%%%%%%%%%%%%%%%%%%%%%%%%%%%%

\section{Introduction}

Since the introduction of the standard circuit model \cite{NCbook},
many new models of quantum computation \cite{model:measurement,model:graph,model:anyons,model:adiabatic,model:walk} 
have been developed. Inspired by Feynman's ideas about how to use quantum systems to simulate reversible computation
\cite{Feynman:85}, a class of `analog' Hamiltonian quantum computing models computing is based on time evolution with time-independent or slowly changing Hamiltonians.
%These models brought forth another interesting vantage point from which to examine
%the usual quantum circuit model composed of a sequence of unitaries.
The basic idea is to encode the progress of a quantum computation $U$ (a quantum circuit composed of a sequence of local unitaries) with $L$ gates 
\begin{eqnarray}
	U = U_L \cdots U_2 U_1
	\label{Ucircuit}
\end{eqnarray}
on an initial state $\ket{\varphi_0}_w$ into a {\em history state} of a system with two registers $\cH_{work}\otimes \cH_{clock}$ as
\begin{eqnarray}
	\ket{\psi_{hist}} 
		&=& \frac{1}{\sqrt{L+1}} \sum_{t=0}^{L} 
		\underbrace{\big(U_t \dots U_0 \ket{\varphi_0}_{w}\big)}_{\ket{\varphi_t}_w}
		\otimes \ket{t}_{c}. \label{history}
\end{eqnarray}

First, one can now look for a simple Hamiltonian with \eqref{history} as its ground state. Finding the ground state of such a Hamiltonian, measuring its clock register and 
obtaining the result $L$ would then allow us to obtain $\ket{\varphi_L}$, the result 
of the quantum circuit $U$ applied to the initial state $\ket{\varphi_0}$ of the work register. This approach produced many results about the quantum complexity 
(in particular, $\QMA$-completeness) of finding the ground state properties of simple Hamiltonians \cite{Kitaev:book,lh:KR03,lh:KKR06,lh:OT08,Nagaj:07a,GottesmanLine:07,ERtriangle}.

A second direction of investigation in these models \cite{JW:05, VC-QCA:07, NW:08, AQC:Mizel} is to look at the Hamiltonian dynamics instead of the static ground state properties. The Hamiltonians in these {\em Hamiltonian Quantum Computing} (HQC) models 
generate Schr\"odinger time evolution which brings an easily-prepared initial state towards the history state \eqref{history} or its variants. 
As we show in Section \ref{clocksection}, one can extract the result of the computation $U\ket{\varphi_0}$ from the final state with high probability. This results in $\BQP$ universality -- any problem which is efficiently solvable on a quantum computer
can be solved in the HQC model as well.
%If the Hamiltonian of a HQC is translationally invariant \cite{NW:08}, the model can be called a Hamiltonian quantum cellular automaton. 
Another example utilizing the preparation of the state \eqref{history} is the result of Aharonov et. al. \cite{Aharonov:07a}, 
showing that Adiabatic Quantum Computation 
\cite{model:adiabatic} is polynomially equivalent to the standard circuit model. 

In this paper, we take the second point of view, focusing on the universality of the dynamical properties of a class of Hamiltonians. The desirable properties for such a HQC model are low locality (a Hamiltonian composed of terms which act nontrivially on a few qubits at a time), a simple geometry of interactions, low dimensionality of particles involved (qubits or qudits), time invariance (or slow change with time), translational-invariance and a large eigenvalue gap above the ground state energy
(or above the computational subspace) to protect the computation against decoherence.
We achieve this last property by using a {\em frustration-free} Hamiltonian
-- a sum of positive semidefinite terms, whose ground state is the common ground state of all the terms.

Our goal is to find a universal HQC with a low locality of interactions, using a Hamiltonian that is time-independent and frustration-free.
Feynman's original construction requires 4-local terms (see Section \ref{clocksection} for details). We equip our quantum computer with a different, {\em railroad-switch} clock register. This results in an effective translation from the circuit model to a Hamiltonian (continuous-time) model of quantum computation with a simple, frustration-free, time-independent, now only 3-local Hamiltonian. Note, that to show the universality of our model, we do not need to use the adiabatic theorem, because our Hamiltonian is frustration-free (see Section \ref{AQCsection}).

Recently, Mizel et. al. \cite{AQC:Mizel} took an alternative approach with a construction using physical particles jumping down a grid, utilizing only 2-particle (but 4-site) interactions. For comparison, when we look at their construction encoded using qubits, the interactions become 4-local. 

After finishing the present article, we found out that \cite{Feynman:85} also contains a 3-local version of this computer, using a so-called {\em switch} primitive, unknown to the wider community. The details of the dynamics of Feynman's model with such a clock were then analyzed by DeFalco et al. in \cite{Falco1,Falco2}, later focusing on the entropy of the clock register in \cite{Falco3,Falco4}. Nevertheless, there are two things that make our re-discovery different from previous work. First, we 
in the 3-local construction we focus on choosing a random time for the final measurement, which requires less detailed analysis than the one in \cite{Falco1}. 
We also look at the problem from a modern quantum information perspective, and make connections of our work to other universality models such as adiabatic quantum computation. Second, we take a step further and make the clock construction a two-local one (for a qubit and a qutrit). This paves the way for a possible future new proof of universality of 2-local Hamiltonians without the requirement of large norm penalty terms dictated by perturbation theory gadgets as in \cite{lh:KKR06}.

The paper is organized as follows. First, in Section \ref{clocksection} we review the continuous-time quantum computing model based on Feynman's idea and the usual ways of encoding the clock register. We present our main innovation, the clock register with railroad-switch gadgets in Section \ref{switchsection}, and bound the required running time of the resulting 3-local HQC model in Section \ref{dynamicssection}. Next, we discuss how our result relates to the universality of Adiabatic Quantum Computation in Section \ref{AQCsection} and analyze the gap protecting the computational subspace in Section \ref{gapsection}.
Section \ref{section23} contains our second result, a universal HQC model involving only 2-local qubit-qutrit interactions.
Finally, in Appendix \ref{walksection} we prove a Lemma about the mixing of a 
continuous-time quantum walk on a cycle, utilized in Section \ref{dynamicssection} for the analysis of the required running time of our model.

%%%%%%%%%%%%%%%%%%%%%%%%%%%%%%%%%%%%%%%%%%%%%%%%%%%%%%%%%%%%%%%%%%
%%%%%%%%%%%%%%%%%%%%%%%%%%%%%%%%%%%%%%%%%%%%%%%%%%%%%%%%%%%%%%%%%%

\section{Feynman's Computer and Clock Constructions}
\label{clocksection}

\begin{figure}
	\begin{center}
	\includegraphics[width=4in]{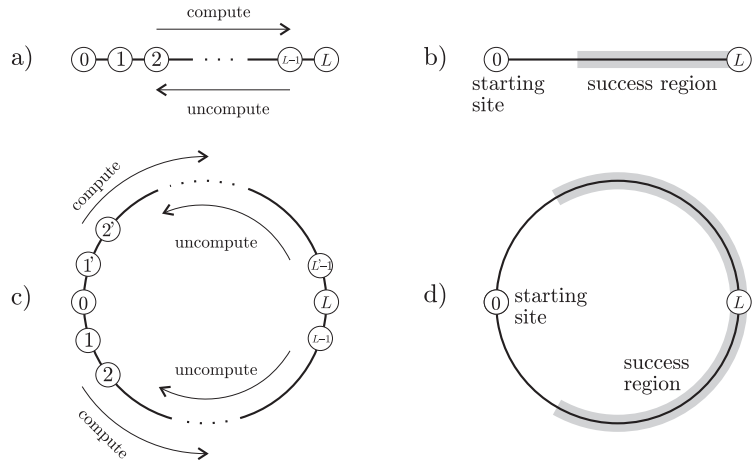} 
	\end{center}
	\caption{a) Computation on a ``line'' of states $\ket{\psi_t}$. b)
	The line corresponding to a circuit $U$ padded with extra identity gates. The region where the computation is done is marked. c) Computation on a cycle of states $\ket{\psi_t}$ coming from two copies of the computation on a line. d) The cycle corresponding to a padded circuit $U$.
	\label{figurecircle}}
\end{figure}
Already in 1985, Feynman \cite{Feynman:85} investigated the possibility of simulating a sequence of reversible classical gates in a quantum-mechanical system. The construction he found applies equally well to implementing a quantum circuit, which is also reversible. Let $U=U_L\cdots U_1$ \eqref{Ucircuit} be a sequence of $L$ unitary transformations $U_t$ on $n$ qubits. Consider a quantum system with two registers, $\cH_{work}\otimes\cH_{clock}$, where the first of these holds $n$ work qubits $q_1,\dots,q_n$. 
The role of the clock register is to hold a pointer - indicating how far the computation has progressed. 
A simple time-independent Hamiltonian then facilitates the evaluation of the circuit $U$ on the work register of this system. Consider the Hamiltonian
\begin{eqnarray}
	H_{F} =  \sum_{t=1}^{L} \left(
		   U_t \otimes \ket{t}\bra{t-1}_{c}
		+  U_t^\dagger \otimes \ket{t-1}\bra{t}_{c} \right),
	\label{HFeynman1}
\end{eqnarray}
where $U_t$ acts on the work register.
The first type of transition given by this Hamiltonian increases the clock register from $\ket{t-1}_c$ to $\ket{t}_c$, while applying the unitary $U_t$ to the corresponding work qubits. 
Similarly, the other transition decreases the clock register and uncomputes the unitary $U_t$.
Fixing an initial state $\ket{\varphi_0}_w$ of the work qubits
and taking $t=0,\dots,L$, we can define a basis
\begin{eqnarray}
	\ket{\psi_t} = \underbrace{\big(U_t \dots U_2 U_1 \ket{\varphi_0}_w\big)}_{\ket{\varphi_t}_w} \otimes \ket{t}_{c}
	%= \ket{\varphi_t}_w \otimes \ket{t}_{c}.
	\label{psit}
\end{eqnarray}
of a particular subspace $\cH_{\varphi_0}$ of the whole Hilbert space $\cH_{work}\otimes\cH_{clock}$. 
Within each such subspace, Feynman's Hamiltonian has a very simple form:
\begin{eqnarray}
	H_{F}\big|_{\cH_{\varphi_0}} = 
	 \left[\begin{array}{rrrrrr}
		0 	& 1	& 		&  	&		&					\\
		1 	& 0	&	1 	&  	&		&					\\
				& 1	&   	&   &		&					\\
		 		& 	&  	& \ddots	&   &					\\
		 		&	  & 		&		&  	&	1	\\
		 		&	  & 		&	  & 1 &	0	
	\end{array}\right].
	\label{nicebasis}
\end{eqnarray}
It is a Hamiltonian of a quantum walk on the ``line'' of states 
$\ket{\psi_t}$, as each state $\ket{\psi_t}$ is connected by a transition in $H_F$ only 
to its direct neighbors (see Figure \ref{figurecircle}).
%We initialize the system in the state 
%$\ket{\psi_0} = \ket{\varphi_0}_w \otimes \ket{0}_{c}$
%where $\ket{\varphi_0}_w$ is an initial state of the work qubits. 
Time evolution with $H_{F}$ according to the Schr\"odinger equation thus brings the initial state $\ket{\psi_0} = \ket{\varphi_0}_w \otimes \ket{0}_c$ into some superposition of the states $\ket{\psi_t}$ \eqref{psit}.
After some time $\tau$, let us measure the clock register. When we obtain $L$, the state of the system after the measurement becomes $\ket{\psi_L} = \ket{\varphi_L}_w \otimes \ket{L}_{c}$. The work register then contains the desired output of the circuit $U$. 
One might argue that this doesn't happen too often. However, we can boost the probability of obtaining $\ket{\psi_L}$ by slightly modifying our system. First, we pad the circuit with $2L$ extra identity operations
\begin{eqnarray}
	U' = \ii_{3L} \ii_{3L-1} \dots \ii_{L+1} U,
	\label{padding}
\end{eqnarray}
and, accordingly, expand the clock register to hold states $\ket{0}_c,\dots,\ket{3L}_c$ (see Figure \ref{figurecircle}b).
All the states $\ket{\varphi_t}_w$ of the work register with $t\geq L$ are now the same -- equal to the desired output state $\ket{\varphi_L}_w$.
We now let the system evolve from $\ket{\psi_0}$ for a time $\tau$ chosen uniformly at random between $0$ and $O(L\log^2 L)$. Because the quantum walk on a line is rapidly mixing, the probability to measure a state with the clock register $\ket{t\geq L}_{c}$ at time $\tau$ is close to $\frac{2}{3}$, as we show in Appendix \ref{walksection}. This yields the desired state $\ket{\varphi_L}_w$ of the work qubits. If a measurement of the clock register gives us $t<L$, we repeat the experiment.

Note also that we can turn the Feynman Hamiltonian \eqref{HFeynman1} into a sum of projector terms. To do this, we change the sign of $H_F$ and add terms of the type $\ket{t-1}\bra{t-1}_c$ and $\ket{t}\bra{t}_c$ as 
\begin{eqnarray}
		H_{F}^{proj} 
		%= \ii - \frac{1}{2}\left(\ket{0}\bra{0}_c + \ket{L}\bra{L}_c\right)
		%- \frac{1}{2} H_F
		= \sum_{t=1}^{L} H_{proj}^{(t)} 
		=  \sum_{t=1}^{L} \frac{1}{2}\left( \ket{t-1}\bra{t-1}_c + \ket{t}\bra{t}_c - 
		   U_t \otimes \ket{t}\bra{t-1}_{c}
		- U_t^\dagger \otimes \ket{t-1}\bra{t}_{c} \right).
	\label{HFeynmanP}
\end{eqnarray}
Just as in \eqref{nicebasis},
the Hamiltonian $H_{F}^{proj}$
restricted to a subspace spanned by the basis 
\eqref{psit} has a simple matrix form
\begin{eqnarray}
	H_{F}^{proj}\big|_{\varphi_0} &=& 
	\half
	\left[\begin{array}{rrrrrr}
		1 	& -1	& 				&  				&					&					\\
		-1 	& 2				&	-1 	&  				&					&					\\
						& -1	& 2      	& -1  &					&					\\
		 				& 			  & -1 	& \ddots	& \ddots  &					\\
		 				& 			  &  				&	\ddots	& 2			 	&	-1	\\
		 				& 			  &  				&	      	& -1 	&	1		
	\end{array}\right].
	\label{HFeynmanPmatrix}
\end{eqnarray}
%Note that this matrix is the same as the Hamiltonian of $L+1$ spins on a line with interaction 
%\begin{eqnarray}
%	H_{\frac{1}{2}} &=& \sum_{j=1}^{L} \frac{1}{2} \left(
%	 \ket{00}\bra{00}+\ket{11}\bra{11}-\ket{10}{01}-\ket{01}\bra{10}\right).
%\end{eqnarray}
Observe that as the rows sum to zero, 
for any initial state of the work qubits $\ket{\varphi_0}_w$, the corresponding history state $\frac{1}{\sqrt{L+1}}\sum_{t=0}^{L}\ket{\psi_t}$ \eqref{history} is the ground state of $H_F^{proj}$, with energy equal to zero. Again, $H_{F}^{proj}$ is a Hamiltonian of a quantum walk on a line,
although now with a small modification at the endpoints and a shifted spectrum:
\begin{eqnarray}
		H_{F}^{proj} 
		= \ii 
		- \frac{1}{2} H_F
		- \frac{1}{2}\left(\ket{0}\bra{0}_c + \ket{L}\bra{L}_c\right).
\end{eqnarray}
This is the form of the Hamiltonian which Kitaev used in \cite{Kitaev:book} (along with additional terms) to prove $\QMA$-completeness of the Local Hamiltonian problem.
Finally, note that we can get rid of the endpoint terms by wrapping the computation around on a 
circle as in Figure \ref{figurecircle}c.

\subsection{Implementing the Clock Register}
\label{implementclocks}

So far, we haven't said anything about how the clock register is actually implemented. 
Its states should indicate where in the computation we are, encoding the $L+1$ different states $\ket{t}_c$. The first na\"ive approach is to use a binary encoding of $t$ which would require $\log (L+1)$ qubits. However, the transition operators $\ket{t}\bra{t-1}_c$ in the Feynman Hamiltonian would then be highly nonlocal. It is much better to use the idea Feynman originally had in mind. Imagine the clock register as the state space of a pointer particle hopping on a line $0,\dots,L$, with the state $\ket{t}_c$ corresponding to the pointer positioned at $t$. This {\em pulse clock} uses $L+1$ clock qubits $c_0,\dots,c_L$, representing the states $\ket{t}_c$ as
\begin{eqnarray}
	\ket{0}_c &=& \ket{\mathtt{1000\dots 00}}, \label{pulseclock} \\
	\ket{1}_c &=& \ket{\mathtt{0100\dots 00}}, \nonumber \\
	\ket{2}_c &=& \ket{\mathtt{0010\dots 00}}, \nonumber \\
	&\vdots& \nonumber\\
	\ket{L}_c &=& \ket{\mathtt{0000\dots 01}}, \nonumber
\end{eqnarray}
with a single clock qubit (indicating the pointer) in the state $\ket{\mathtt{1}}$. 
What is the reason to take such a generous approach, using only $L+1$ of the
$2^{L+1}$ states available in the state space of the pulse clock register?
For the pulse clock, the transition operators in \eqref{HFeynman1} acting on the clock register can be implemented just 2-locally, acting nontrivially on only two clock qubits and as an identity on the rest of the system:
\begin{eqnarray}
	\ket{t}\bra{t-1}_{c} = \ii%_{c_0,\dots,c_{t-1}} 
													\otimes \ket{\mathtt{01}}\bra{\mathtt{10}}_{c_{t-1},c_{t}}
													\otimes \ii%_{c_{t+1},\dots,c_{L}}
													.
	\label{transitionpulse}
\end{eqnarray}
Two-qubit unitaries $U_t$ are necessary for universal quantum computation, so the Hamiltonian \eqref{HFeynman1} with the pulse clock encoding contains 4-local terms. 
Our main result, given in Section \ref{switchsection}, uses a pulse-type clock and implements \eqref{HFeynman1} just 3-locally.
Note though, that there are two problems with the pulse clock. First, it is not entirely simple to locally check the clock register for correctness (ruling out the states with two neighboring $\mathtt{1}$'s, but also with two far-away $\mathtt{1}$'s). Second, the pulse clock requires initialization -- ensuring that the clock register {\em has} a pointer (a spin in the state $\ket{\mathtt{1}}$) at all. This can't be done by using only local positive semidefinite terms in the Hamiltonian, which makes the pulse clock unsuitable for obtaining $\QMA_1$-completeness results such as %\cite{Kitaev:book,lh:KR03,lh:KKR06,lh:OT08,Nagaj:07a,GottesmanLine:07,ERtriangle}.
\cite{Qksat,Nagaj:07a,ERtriangle}.

Another way of encoding the states $\ket{t}_c$ is a {\em domain wall clock}, using a unary representation of $\ket{t}_c$, introduced by Kitaev \cite{Kitaev:book} for his Local Hamiltonian problem. 
The progression of states (on $L+2$ clock qubits) is
\begin{eqnarray}
	\ket{0}_c &=& \ket{\mathtt{1000\dots 000}}, \label{wallclock} \\
	\ket{1}_c &=& \ket{\mathtt{1100\dots 000}}, \nonumber \\
	\ket{2}_c &=& \ket{\mathtt{1110\dots 000}}, \nonumber \\
	&\vdots& \nonumber\\
	\ket{L}_c &=& \ket{\mathtt{1111\dots 110}}, \nonumber
\end{eqnarray}
where the ``clock time'' (or alternatively, pointer position) $t$ is determined by the position of the domain wall $\mathtt{10}$.
A useful property of the domain wall clock is that we can check whether the state of the clock register is in one of the states $\ket{t}_c$ \eqref{wallclock} by 2-local operators. What suffices is to ensure that the sequence $\mathtt{01}$ never occurs. Moreover, we can ensure that there is a single domain wall $\mathtt{10}$ in the clock register by imposing that the first qubit is fixed to $\mathtt{1}$ and that the last qubit is $\mathtt{0}$.
This solves the initialization problem of the pulse clock (where we can't rule out the state $\ket{\mathtt{00\cdots00}}_c$ by simultaneously satisfiable local positive-semidefinite operators).
With the domain wall clock, the transition operators in $H_F$ \eqref{HFeynman1} are 3-local, as we need 
\begin{eqnarray}
	\ket{t}\bra{t-1}_{c} = \ii%_{c_0,\dots,c_{t-1}} 
													\otimes \ket{\mathtt{110}}\bra{\mathtt{100}}_{c_{t-1},c_{t},c_{t+1}}
													\otimes \ii%_{c_{t+2},\dots,c_{L}}
													.
\end{eqnarray}
Recall that we want to increase the clock register while appling two-qubit unitaries. Feynman's Hamiltonian with the domain wall clock is thus 5-local.
If one uses transition operators with smaller locality (as in \cite{lh:KR03,lh:KKR06}), transitions to illegal clock states (states of the clock register other than the prescribed $\ket{t}_c$) will occur in the Hamiltonian time evolution with $H_F$. This needs to be dealt with by imposing large energy penalty terms in the Hamiltonian. Moreover, the Hamiltonian will no longer be frustration-free.
For this reason, the domain wall clock does not fit our present purpose: finding a simple, low-locality, frustration-free Hamiltonian with a small norm, allowing us to do universal quantum computation by the time-evolution it generates.

Another approach is to combine the two clock constructions above by interspersing a domain wall clock lattice with a pulse clock lattice. In this fashion, Nagaj and Mozes \cite{Nagaj:07a} and Eldar et al. \cite{ERtriangle} showed $\QMA$-completeness of the 3-local Hamiltonian problem \cite{footnoteLH} with a much larger promise gap.

Going beyond the linear structure of the pulse/domain wall clocks, Eldar et al. \cite{ERtriangle} have recently come up with a novel idea. They use a progression of states of the clock register that can go from $\ket{t-1}_c$ to $\ket{t}_c$ by two different routes, depending on the state of one of the work qubits. This results in their proof that the Quantum 2-SAT problem \cite{Qksat} for a cinquit-qutrit pair (each Q-2-SAT condition acts on one 5-dimensional and one 3-dimensional quantum system) is $\QMA_1$-complete \cite{ERtriangle}. In the next section, we take the nonlinear route as well, construct a novel, pulse-type clock and 
make the Feynman Hamiltonian 3-local (as opposed to 4-local with the usual pulse clock). What makes our construction different from \cite{ERtriangle} is a bit-flip symmetry which allows us to easily analyze the system's dynamics and prove its universality.

%%%%%%%%%%%%%%%%%%%%%%%%%%%%%%%%%%%%%%%%%%%%%%%%%%%%%%%%%%%%%%%%%%%%%%%%%%%%%%%%%%%%%%%%%%
%%%%%%%%%%%%%%%%%%%%%%%%%%%%%%%%%%%%%%%%%%%%%%%%%%%%%%%%%%%%%%%%%%%%%%%%%%%%%%%%%%%%%%%%%%

\section{The Railroad Switch: a 3-local Gadget}
\label{switchsection}

For a given quantum circuit $U$, our goal is to construct a Hamiltonian which facilitates time evolution
of a particular initial state towards a state very much like \eqref{history}, encoding the result of $U$. Specifically, we focus on the clock register encoding of the Feynman Hamiltonian \eqref{HFeynman1}.
We base our construction on a pulse clock \eqref{pulseclock}. We can view the progression of the pointer $\mathtt{1}$ in the clock states $\ket{t}_c$ as a train moving along a single track as in Figure \ref{figurepulse}. As it goes between stations $c_{t-1}$ to $c_t$, one-qubit or two-qubit unitary operations $U_t$ are applied to the work qubits.
\begin{figure}
	\begin{center}
	\includegraphics[width=2.5in]{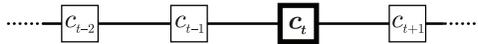} 
	\end{center}
	\caption{The {\em pulse clock} is a line (a single track) of $L+1$ clock qubits with a single active qubit (the train) in the state $\ket{\mathtt{1}}$. When the clock qubit $c_{t}$ is active,
	%(and all other clock qubits are inactive), 
	the pulse clock encodes the state $\ket{t}_c$.
	\label{figurepulse}}
\end{figure}
For any single qubit gate, the operator $U_{t}\otimes \ket{t}\bra{t-1}_c$ is already 3-local \eqref{transitionpulse}.
Besides single-qubit gates, all that is necessary for universality is the ability to apply the 2-qubit CNOT gate. 
As we have seen in the previous Section, 4-local interactions are required for this in the standard Feynman Hamiltonian construction.
To do it with only 3-local interactions, we have to add a twist to the pulse clock,
going beyond its linear structure. 
For each CNOT gate, instead of viewing the clock register as a train on a single track as in Figure \ref{figurepulse}, we introduce
a railroad switch gadget with two tracks, depicted in Figure \ref{figureswitch}.
The train will use the upper or lower track depending on the state of a `train master' $q_1$ -- one of the work qubits. Only on the upper track, the target work qubit $q_2$ is flipped. The gadget thus facilitates the application of a CNOT gate on the work qubits $q_1$ and $q_2$.
\begin{figure}
	\begin{center}
	\includegraphics[width=3.6in]{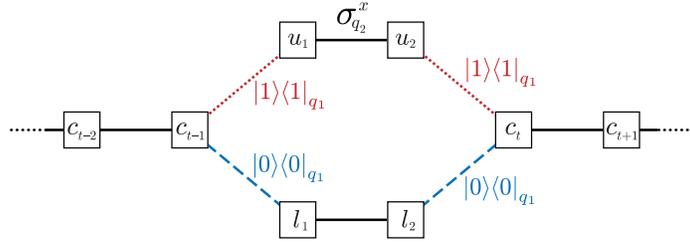} 
	\end{center}
	\caption{The {\em railroad switch} gadget for the application of a CNOT gate between work qubits $q_1$ and $q_2$, using only 3-local terms. 
	Comparing to the pulse clock in Figure \ref{figurepulse}, there are four extra clock qubits inserted between $c_{t-1}$ and $c_{t}$, creating two tracks. The routing of the train (active site) from $c_{t-1}$ (and $c_{t}$) to the upper and lower tracks is conditioned on the state of the train master -- a control qubit $q_1$ in the work register. Furthermore, when the train moves from $u_1$ to $u_2$ on the upper track, the target work qubit $q_2$ is flipped. 
	\label{figureswitch}}
\end{figure}

Our resulting Hamiltonian $H_{3S}$ is composed from positive semidefinite 3-local terms,
so it is a Quantum 3-SAT Hamiltonian. 
We prove the universality of this 3-local Hamiltonian Computer model where we let an easily prepared starting state evolve with $H_{3S}$ for not too long a time, and also for an Adiabatic Quantum Computation (AQC) where we slowly change the Hamiltonian from a simple starting Hamiltonian $H_B$ to $H_{3S}$. We discuss this connection to AQC in detail in Section \ref{AQCsection}.

Let us present the details of the railroad switch. The original pulse clock has $L+1$ qubits $c_0,\dots,c_L$. When a CNOT gate is due at step $t$ (i.e. $U_t$=CNOT), we introduce four extra clock qubits ($u_1,u_2,l_1,l_2$) between
$c_{t-1}$ and $c_t$ as in Figure \ref{figureswitch}.
Let us call the control work qubit $q_1$, and the target qubit for the CNOT $q_2$.
If $q_1$ is in the state $\ket{1}$, the active site (the train) in the clock register can only follow the upper track. For $q_1$ in the state $\ket{0}$, it must take the lower track. 
Furthermore, when the train moves from $u_1$ to $u_2$, we flip the target qubit $q_2$. The train then moves to the clock qubit $c_{t}$ and the two tracks (computational paths) merge. 
This effectively applies a CNOT gate between $q_1$ and $q_2$. 
The transition operators for the upper track are now all 3-local:
\begin{eqnarray}
	\ket{1}\bra{1}_{q_1} &\otimes& 
 			\big(\ket{u_1}\bra{c_{t-1}} + \ket{c_{t-1}}\bra{u_1}\big), \label{Hu1}\\ 
	 \hskip22pt \sigma_{q_2}^{x} &\otimes& 
			\big( \ket{u_2}\bra{u_1} + \ket{u_1}\bra{u_2} \big), \label{Hu2}\\
	 \ket{1}\bra{1}_{q_1} &\otimes& 
 			\big(\ket{c_t}\bra{u_2} + \ket{u_2}\bra{c_t}\big), \label{Hu3}
\end{eqnarray}
because the operators $\ket{u_1}\bra{c_t}$, etc. can be implemented for our pulse-type clock only 2-locally as
\begin{eqnarray}
	\ket{u_1}\bra{c_{t-1}} &=& \ket{\mathtt{10}}\bra{\mathtt{01}}_{c_{t-1},u_1}, \\ 
	\ket{u_1}\bra{u_2} &=& \ket{\mathtt{10}}\bra{\mathtt{01}}_{u_1,u_2}, \\ 
	\ket{u_2}\bra{c_{t}} &=& \ket{\mathtt{10}}\bra{\mathtt{01}}_{u_2,c_{t}}, 
\end{eqnarray}
acting trivially on the rest of the clock qubits. 
The Hamiltonian for the upper track is then a sum of \eqref{Hu1}-\eqref{Hu3}. If we desire so, we have the option of making it a sum of projectors by adding a few terms as in \eqref{HFeynmanP}. The projector version of the Hamiltonian for the upper track is then a sum of three 3-local projectors:
%\begin{eqnarray}
%	H_{up}^{(t)} &=& \ket{1}\bra{1}_{q_1} \otimes \frac{1}{2} 
%		\big(
%		\ket{c_{t-1}}\bra{c_{t-1}} + \ket{u_1}\bra{u_1} 
%		- \ket{u_1}\bra{c_{t-1}} - \ket{c_{t-1}}\bra{u_1}
%		\big) \\ 
%			&+&  \phantom{\ket{1}\bra{1}_{q_1} \otimes } \,
%			\frac{1}{2} \big(
%		\ket{u_1}\bra{u_1} + \ket{u_2}\bra{u_2} 
%		- \sigma_{q_2}^{x} \otimes \ket{u_2}\bra{u_1} - \sigma_{q_2}^{x} \otimes  \ket{u_1}\bra{u_2}  
%		\big) \nonumber \\
%	&+&  \ket{1}\bra{1}_{q_1} \otimes \frac{1}{2} 
%		\big(
%		\ket{u_2}\bra{u_2} + \ket{c_{t}}\bra{c_{t}} 
%		- \ket{u_2}\bra{c_{t}} - \ket{c_{t}}\bra{u_2}
%		\big). \nonumber 
%\end{eqnarray}
\begin{eqnarray}
	H_{up}^{(t)} &=& \ket{1}\bra{1}_{q_1} \otimes \frac{1}{2} 
		\big(
		\ket{c_{t-1}}\bra{c_{t-1}} + \ket{u_1}\bra{u_1} 
		- \ket{u_1}\bra{c_{t-1}} - \ket{c_{t-1}}\bra{u_1}
		\big) \\ 
			&+&  \phantom{\ket{1}\bra{1}_{q_1} \otimes } \,
			\frac{1}{2}\, 
			\ii_{q_2}\otimes
			\big(
		\ket{u_1}\bra{u_1} + \ket{u_2}\bra{u_2} 
		\big)
		-\frac{1}{2}\,\sigma_{q_2}^{x}\otimes\big(
		\ket{u_2}\bra{u_1} - \ket{u_1}\bra{u_2}  
		\big) \nonumber \\
	&+&  \ket{1}\bra{1}_{q_1} \otimes \frac{1}{2} 
		\big(
		\ket{u_2}\bra{u_2} + \ket{c_{t}}\bra{c_{t}} 
		- \ket{u_2}\bra{c_{t}} - \ket{c_{t}}\bra{u_2}
		\big). \nonumber 
\end{eqnarray}
Note that when the clock register is in a proper clock state, i.e. one and only one clock qubit is in the state $\ket{\mathtt{1}}$, the projectors such as $\ket{c_t}\bra{c_t}$ are easily implemented 1-locally as $\ket{\mathtt{1}}\bra{\mathtt{1}}_{c_t}$.

Analogously, the transition operators for the lower track are
\begin{eqnarray}
	\ket{0}\bra{0}_{q_1} &\otimes&
 			\big(\ket{l_1}\bra{c_{t-1}} + \ket{c_{t-1}}\bra{l_1}\big), \label{Hl1}\\ 
	\ii_{q_2} &\otimes&
			\big( \ket{l_2}\bra{l_1} + \ket{l_1}\bra{l_2} \big), \label{Hl2}\\
	\ket{0}\bra{0}_{q_1} &\otimes&
 			\big(\ket{c_t}\bra{l_2} + \ket{l_2}\bra{c_t}\big). \label{Hl3}
\end{eqnarray}
Again, we can add a few terms and turn them into projectors, obtaining
\begin{eqnarray}
	H_{lo}^{(t)} &=&  \ket{0}\bra{0}_{q_1} \otimes \frac{1}{2} 
 			\big(
 			\ket{l_1}\bra{l_1} + \ket{c_{t-1}}\bra{c_{t-1}}
 			-	\ket{l_1}\bra{c_{t-1}} - \ket{c_{t-1}}\bra{l_1}
 			\big), \\ 
	&+& \phantom{\ket{0}\bra{0}_{q_1} \otimes} \,
		 \frac{1}{2} 
			\big( 
			\ket{l_1}\bra{l_1} + \ket{l_2}\bra{l_2}
			- \ket{l_2}\bra{l_1} - \ket{l_1}\bra{l_2} 
			\big), \nonumber \\
	&+& \ket{0}\bra{0}_{q_1} \otimes \frac{1}{2} 
	 			\big(
	 			\ket{l_2}\bra{l_2} + \ket{c_t}\bra{c_t}
	 			- \ket{c_t}\bra{l_2} - \ket{l_2}\bra{c_t}
	 			\big). \nonumber
\end{eqnarray}

We now replace each term $H_{proj}^{(t)}$ in \eqref{HFeynmanP}
for those $t$ where $U_t=$CNOT with the railway switch gadget Hamiltonian
\begin{eqnarray}
	H_{switch}^{(t)} = H_{up}^{(t)} + H_{lo}^{(t)},
	\label{Hrail}
\end{eqnarray}
obtaining 
\begin{eqnarray}
	H_{3S} &=& 
			\sum_{t: \textrm{1-qubit }U_t} H_{proj}^{(t)}
			+ \sum_{\textrm{CNOTs}} H_{switch}^{(t)}
		\label{h3s}
\end{eqnarray}

%%%%%%%%%%%%%%%%%%%%%%%%%%%%%%%%%%%%%%%%%%%%%%%%%%%%%%%%%%%%%%%%%%%%%%%%%%%%%%%%%%%%%%%%%%
%%%%%%%%%%%%%%%%%%%%%%%%%%%%%%%%%%%%%%%%%%%%%%%%%%%%%%%%%%%%%%%%%%%%%%%%%%%%%%%%%%%%%%%%%%

\subsection{The Dynamics of our Model}
\label{dynamicssection}

Let us now look at the time evolution generated by this Hamiltonian, focusing first on a single
railway switch. 
Denote the work qubits on which the CNOT gate acts $q_1$ and $q_2$ .
Consider the action of \eqref{Hrail} on the state 
\begin{eqnarray}
	\ket{\psi_{t-1}} 
		&=& \kets{\varphi_{t-1}}_w \otimes \ket{t-1}_c = \big(
		a \underbrace{\ket{0}_{q_1} \ket{\alpha}_{q_2,\dots}}_{\ket{\varphi_{0\dots}}} 
		+ b \underbrace{\ket{1}_{q_1} \ket{\beta}_{q_2,\dots}}_{\ket{\varphi_{1\dots}}}
		\big)
		\otimes \ket{c_{t-1}}_c,
\end{eqnarray}
where $\kets{\varphi_{t-1}} = a \ket{\varphi_{0\dots}} + b \ket{\varphi_{1\dots}}$ is a general state of the work register
and the clock register has the active site (train) at $c_{t-1}$.
While the train is within the railroad switch gadget, the transitions given in \eqref{Hrail} dictate that the state of the system does not leave the subspace spanned by the following four states:
\begin{eqnarray}
	\ket{\psi_{t-1}} 
		&=& \big(
		a \kets{\varphi_{0\dots}} + b \kets{\varphi_{1\dots}}
		\big)
		\otimes \ket{c_{t-1}}_c, \\
	\kets{\psi^{(1)}_{t-1}} 
	&=& 
		 a \kets{\varphi_{0\dots}}	\otimes \ket{l_1}_c
		+
		b \kets{\varphi_{1\dots}} \otimes \ket{u_1}_c, 
						\label{psi1step}		\\
	\kets{\psi^{(2)}_{t-1}}  
	&=& 
		 a \kets{\varphi_{0\dots}} \otimes \ket{l_2}_c
		+
		 b \kets{\varphi'_{1\dots}} \otimes \ket{u_2}_c, 
		 				\label{psi2step}	\\
	\kets{\psi_{t}} 
	&=& \big(
			 a 
			\kets{\varphi_{0\dots}} + 	b 
			\kets{\varphi'_{1\dots}}
		\big) 
		\otimes \ket{c_t}_c,
\end{eqnarray}
where $\kets{\varphi'_{1\dots}} = \sigma_{q_2}^{(x)} \ket{\varphi_{1\dots}}$. Observe that the work qubits of the state $\kets{\psi_{t}}$ have the CNOT gate applied to $q_1,q_2$, as we wanted, meaning that
\begin{eqnarray}
	\kets{\psi_{t}} 
	&=& 
		\textrm{CNOT}_{q_1,q_2}\ket{\varphi_{t-1}}_w
	\otimes\ket{t}_c.
\end{eqnarray}
This gives us a method of applying a 2-local CNOT gate while updating the clock register at the same time -- via only 3-local interactions between the work and clock registers.

Before introducing the railroad switch, the progression of states $\dots$,
$\ket{\psi_{t-2}}$, $\ket{\psi_{t-1}}$, $\kets{\psi_{t}}$, $\ket{\psi_{t+1}}$, $\dots$ was a ``line'', as 
each state in it was connected only to its two nearest neighbors by a transition in $H$. Now, 
with the transitions given by \eqref{h3s}, 
the modified progression of states $\dots$, $\ket{\psi_{t-1}}$, $\kets{\psi_{t-1}^{(1)}}$, $\kets{\psi_{t-1}^{(2)}}$, $\ket{\psi_t}$, $\dots$ again forms a ``line''. 
Let us call 
\begin{eqnarray}
	\cH_{legal} = span\{ \ket{\psi_t} \}
	\label{legal}
\end{eqnarray}
the subspace of $\cH$ spanned by this new progression of states, including those of the form \eqref{psi1step} and \eqref{psi2step}. Because we constructed it so, $H_{3S}$ \eqref{h3s} does not induce transitions between the subspace $\cH_{legal}$ and $\cH_{legal}^\perp$. The time evolution of an initial state in $\cH_{legal}$ is thus governed solely by the restriction of $H_{3S}$ to $\cH_{legal}$.
Moreover, in the new, updated basis $\ket{\psi_t}$ this restriction has the form \eqref{HFeynmanPmatrix} which we already know. It is a Hamiltonian for a quantum walk on a line of length $L+1 = poly(n)$, where $n$ is the number of qubits the quantum circuit $U$ acts on.
The railroad switch requires three steps for each CNOT gate, so $L=L_1+3L_{CNOT}$, where $L_1$ is the number of single qubit gates and $L_{CNOT}$ is the number of CNOTs in the circuit $U$. 
To avoid complications in the analysis of the required running time coming from the endpoints, we can change the `line' of states $\ket{\psi_t}$ into a circle of length $2L = 2 (L_1+3L_{CNOT})$ as in Figure \ref{figurecircle}c. First, we double the clock register as
\begin{eqnarray}
	\ket{c_0 c_1 \dots c_L} \goes \ket{c_0 c_1 \dots c_L} \otimes \ket{c_{0'} c_{1'} \dots c_{L'}},
\end{eqnarray}
where the Hamiltonian involving the second clock register has the same form \eqref{h3s}. Second, we identify the endpoint qubits, i.e. $c_0 \equiv c_{0'}$ and $c_L \equiv c_{L'}$. This gives the clock register the geometry of a cycle, with a unique state 
\begin{eqnarray}
	\ket{0}_c^{\circ} = \ket{\mathtt{1}_{c_{0}} 
	\begin{array}{c}
		\mathtt{0} \dots \mathtt{0} \\
		\mathtt{0} \dots \mathtt{0} 
	\end{array}
	\mathtt{0}}
	\label{circleclockinit}
\end{eqnarray}
corresponding to $t=0$. The active site (spin up) in the clock register can proceed towards $c_L$ both ways, as in Figure \ref{figurecircle}c.
The span of this new set of states $\ket{\psi_t}$ with the geometry of a circle then defines the subspace $\cH_{legal}^\circ$. 
When we restrict the new Hamiltonian \eqref{h3s} including the terms involving the clock qubits $c_{t'}$
to the subspace $\cH_{legal}^{\circ}$ and express it in the basis $\{\ket{\psi_t}\}$, we get 
\begin{eqnarray}
	H_{3S}\Big|_{\cH_{legal}^\circ} = \half \left(\ii - B_\circ\right),
\end{eqnarray} 
where $B_\circ$ is the adjacency matrix for a cycle of length $2L$. The dynamics of this system is a~quantum walk on a~cycle, which we analyze in detail in Appendix \ref{walksection}. 
There we prove 
\begin{lem2}
Consider a continuous time quantum walk on a cycle of length $L$ (divisible by 4), where the Hamiltonian is the negative of the adjacency matrix for the cycle. Let the system evolve for a time $\tau \leq \tau_{\circ}$ chosen uniformly at random, starting in a position basis state $\ket{0}$.
The probability to measure a position state $\ket{t}$ farther than $L/6$ from the starting point (the farther two thirds of the cycle) is then bounded from below as $p_{\circ}\geq\frac{2}{3}-\frac{1}{3L}-O\left(\frac{L \log^2 L}{\tau_{\circ}}\right)$. 
\end{lem2}
Thus, when starting from the state concentrated on a single site of the line and letting the sytem evolve for a time chosen uniformly at random between zero and a number not larger than $O(L \log^2 L)$, the probability to find the state farther than $L/3$ from the starting point is close to $\frac{2}{3}$.
This corresponds to finding a state with a spin up on clock qubit $c_{t\geq L/3}$ (lower part of the circle) or $c_{t'\geq L/3}$ (upper part of the circle), depicted as the success region in Figure \ref{figurecircle}d. 
In both cases, because the circuit $U$ was padded with identity gates as in \eqref{padding}, the state of the work qubits is then $\ket{\varphi_L}$ \eqref{psit} and contains the output of the quantum circuit $U$. 

%\subsection{The Hamiltonian Quantum Computer}
%\label{HQC}

Therefore, the Hamiltonian \eqref{h3s}
allows us to simulate a quantum computation $U$ with $L_1+L_{CNOT}$ gates (here $L_1$ is the number of single qubit gates and $L_{CNOT}$ counts the CNOT's), as follows:
\begin{enumerate}
\item Construct a pulse clock register with $C = L_1 + L_{CNOT}+1$ qubits. For each CNOT gate, add 4 extra qubits for the railroad switch gadget 
and modify the Hamiltonian as described by \eqref{Hrail}.
\item Pad the clock register with extra $2(L_1 + L_{CNOT})$ qubits and pad the circuit $U$ with identity gates to make it three times longer.
\item Double the clock register and join its endpoints.
\item Initialize the system in the state 
		\begin{eqnarray}
			\ket{\psi_0} = \ket{00\dots 0}_w\otimes \ket{0}_c^{\circ}.
			\label{psi0}
		\end{eqnarray}
\item Let it evolve with $H_{3S}^{\circ}$ (given by \eqref{h3s} including terms for the whole clock cycle) for a time chosen uniformly at random between zero and
$\tau\leq O(L\log^2 L)$. 
\item Measure the clock register qubits in the success region (farther than $L/3$ from the initial site, see Figure \ref{figurecircle}d). With probability close to $\frac{2}{3}$, we will find the active site (spin up) there. The work register now contains the output of the quantum circuit $U$. Restart otherwise.
\end{enumerate}

The Quantum 3-SAT problem (``Does a common zero-energy ground state of a sum of 3-local projectors exist?'') is certainly NP-complete, but we still do not know whether it is also $\QMA_1$-complete \cite{Qksat}. Our railroad switch clock is not suitable for proofs of 
%$\QMA$ and 
$\QMA_1$ hardness of Quantum-$k$-SAT, because it requires initialization (see Section \ref{clocksection}).
However, in this work we found that time-independent Quantum 3-SAT Hamiltonians
can be universal in the following sense. 
We constructed $H_{3S}^{\circ}$: a frustration-free, Quantum 3-SAT Hamiltonian (a sum of $O(L)$ 3-local projector terms) powerful enough to perform universal quantum computation in the Hamiltonian Quantum Computer model based on a quantum walk. When the number of qubits and gates involved in the circuit is $L$, what we need for our Quantum 3-SAT Hamiltonian computer are thus $O(L)$ qubits, a Hamiltonian with $O(L)$ projector terms with norm $O(1)$, resulting in $\norm{H}=O(L)$, and a running time of order $O(L\log^2 L)$. The rescaled required resources (time $\times$ energy) scale as $\tau \cdot \norm{H} = O(L^2 \log^2 L)$.

Furthermore, we show in Section \ref{gapsection} that the computation is protected by an energy gap scaling like $O(L^{-1})$ between the legal subspace $\cH_{legal}^{\circ}$ and the subspaces orthogonal to it, with the exception of the computationally dead subspace
without a train (active site in the clock register). This energy barrier could be increased (at the cost of increasing the locality of interactions) by using error correcting codes like the ones developed for AQC by Jordan et al. \cite{AQC:codes}.

Although local (in the number of particles in each interaction term), Feynman-like Hamiltonians are not geometrically local. Therefore, we would like to continue looking for Hamiltonians with geometrically local interactions (e.g. on a 2D grid), generating time evolutions towards a history state \eqref{history} or its variants, or usable in an AQC algorithm. Oliveira and Terhal \cite{lh:OT08} found such a 2-local Hamiltonian on a 2D grid. The downside of their construction is that the Hamiltonian 
doesn't keep the computational subspace invariant
(it is also not is not frustration-free, with ground states which are history states).
It thus doesn't generate evolution directly towards \eqref{history}. Moreover, because it is based on perturbation gadgets, the eigenvalue gap of the final rescaled Hamiltonian itself is a rather large inverse polynomial in $L$. Its use for AQC is thus impractical. Two of the current ways of approaching this problem are then using higher-dimensional particles (qudits) \cite{NW:08}, or higher locality of interactions \cite{AQC:Mizel}.

%%%%%%%%%%%%%%%%%%%%%%%%%%%%%%%%%%%%%%%%%%%%%%%%%%%%%%%%%%%%%%%%%%%%%%%%%%%%%%%%%%%%%%%%%%
%%%%%%%%%%%%%%%%%%%%%%%%%%%%%%%%%%%%%%%%%%%%%%%%%%%%%%%%%%%%%%%%%%%%%%%%%%%%%%%%%%%%%%%%%%

\subsection{HQC and Adiabatic Quantum Computation}
\label{AQCsection}

The model of computation we just described falls into a general category of Hamiltonian Quantum Computers (HQC) which initialize a system in a simple initial state, let it evolve for not too long a time, and end with a measurement of a few qubits (in the computational basis).
Our HQC requires initializing the system in the state $\ket{\psi_0}$ \eqref{psi0}, and then releasing it to evolve with the Hamiltonian $H_{3S}$. 
We now look at the initialization in more detail and investigate how fast can we peform the `release'. 
For the initialization, it is useful to think of $\ket{\psi_0}$ as the unique ground state
of the 2-local, initial Hamiltonian 
\begin{eqnarray}
	H_{init} = \underbrace{\ket{\mathtt{0}}\bra{\mathtt{0}}_{c_0} + \sum_{j\neq 0} \ket{\mathtt{1}}\bra{{\mathtt{1}}}_{c_j}}_{H_{clockinit}} +
	\underbrace{
	\sum_{k=1}^{n} \ket{1}\bra{1}_{q_k}\otimes \ket{\mathtt{1}}\bra{\mathtt{1}}_{c_0}}_{H_{datainit}}, 
\end{eqnarray}
which is diagonal in the computational basis.
While $H_{clockinit}$ prefers the clock register in the state $\ket{t=0}_c$, the term $H_{datainit}$
gives an energy penalty to states whose work register is not properly initialized, i.e. all zero when the clock is in $\ket{c_0}_c$.
To let the HQC run, we turn off $H_{init}$ and turn on $H_{3S}$. One way to do it is linearly as
\begin{eqnarray}
	H(t\leq T_{init}) = \left(1-\frac{t}{T_{init}}\right) H_{init} + \frac{t}{T_{init}} H_{3S}.
	\label{AQCham}
\end{eqnarray}
%The time $T_{init}$ required for turning off the initial Hamiltonian and turning on the HQC Hamiltonian can be large or small. 
For our choice of $H$ and a large enough $T_{init}$, the conditions of the adiabatic theorem (see e.g. \cite{AQC:Lidar}) could hold. This model would then become an Adiabatic Quantum Computation (AQC) \cite{model:adiabatic}, where the system stays close to the instantaneous ground state of the time-dependent Hamiltonian $H(t)$. At $t=T_{init}$, the system would thus be close to the ground state of $H_{3S}$ -- the history state \eqref{history}. 

Because of the energy gap, the AQC model has some protection against noise during the computation. The robustness and fault tolerance of the AQC has been recently investigated by Childs et al. \cite{AQC:Childs}, Jordan et al. \cite{AQC:codes}, Lidar \cite{AQC:Lidar} and Lloyd \cite{AQC:Lloyd}. The challenge today is to find more new, entirely adiabatic algorithms, such as the adiabatic state preparation of Aharonov and Ta-Shma \cite{AQC:AharonovPrepare}.

In \cite{Aharonov:07a}, Aharonov et al. proved that Adiabatic Quantum Computation is polynomially equivalent to the standard quantum circuit model of quantum computation.
This proof has two directions. First, as shown before by Farhi et. al. \cite{FarhiAQC}
we can simulate time evolution with an AQC Hamiltonian on a standard quantum computer, dividing the time into small intervals and approximating the time evolution in each slice. The usual procedure utilizes the finite $n$ approximation of the Trotter-Suzuki formula and the resulting unitary transformations are applied using a circuit-based quantum computer. Second, \cite{Aharonov:07a} shows that any quantum computation can be simulated by AQC with $T_{init}$ polynomial in $L$ (the number of gates in the circuit $U$). For this, they utilize a Feynman-type Hamiltonian similar to \eqref{HFeynman1}, with a particular encoding of the clock register. The key element is a lower bound on the eigenvalue gap of the Hamiltonian \eqref{AQCham} for all $t$ -- an inverse polynomial in $L$. 

However, in \cite{AQC:Lloyd}, Lloyd has shown that a fast passage from $H_{init}$ to the final Hamiltonian works just as well. Let us consider the extreme case -- an almost instantaneous change in the Hamiltonian. The state of the system then stays close to the initial state. The system then continues to evolve according to the final Hamiltonian. As we have seen in Section \ref{dynamicssection}, this Hamiltonian evolves the initial state within a particular {\em computational subspace} \eqref{legal}
spanned by plane-wave eigenvectors of the Hamiltonian of the quantum walk on the line $\{\ket{\psi_t}, t=0,\dots,L\}$. 
Within this subspace, it rapidly `mixes', 
in the sense that at a random time between $0$ and $O(L \log L)$, the system is likely to have an active site in the clock register in the success region (see Figure \ref{figurecircle}). The sectors of the Hilbert space containing $\ket{\varphi_0}_w\otimes\ket{0}_c$ (computing on the proper initial state) and $\kets{\varphi_0^\perp}_w\otimes\ket{0}_c$ (computing on badly initialized work qubits) do not mix, because the Hamiltonian does not couple them. Running the HQC then can give us the resulting state $U\ket{\varphi_0}$ with high probability, even though we turned on $H_{3S}$ quickly.
The proof of AQC's universality thus works just as well without the adiabatic theorem (requiring some extra time evolution after reaching the final Hamiltonian). 
The essential element required for this result is that the system evolves only within a preferred computational subspace. This in turn is the result of our railroad-switch clock register construction -- which is frustration free, and doesn't induce any `illegal' transitions. Therefore, our result means that previous universality results for 3-local Hamiltonians which required large penalty terms can be greatly simplified,
and are valid even without assuming adiabaticity. We can now ask where the real power of Adiabatic Quantum Computing lies -- e.g. whether the model can stay universal if we only use restricted terms in the Hamiltonian as in \cite{Biamonte}.

%%%%%%%%%%%%%%%%%%%%%%%%%%%%%%%%%%%%%%%%%%%%%%%%%%%%%%%%%%%%%%%%%%%%%%%%%%%%%%%%%%%%%%%%%%
%%%%%%%%%%%%%%%%%%%%%%%%%%%%%%%%%%%%%%%%%%%%%%%%%%%%%%%%%%%%%%%%%%%%%%%%%%%%%%%%%%%%%%%%%%

\subsection{Protecting the Computation}
\label{gapsection}

The Hamiltonian quantum computation model we just presented is based on two premises. First, we initialize the system in a computational basis product state. Second, the Hamiltonian of the system is such that the Schr\"odinger time evolution of this state runs within a particular subspace. Transitions bringing the state of the system out of this computational subspace are not present in the Hamiltonian. This allowed us to understand the dynamics of this model and bound its required running time. We now ask what are the implications of noise on the dynamics of the system and how could we possibly protect the system from moving away from the computational subspace \eqref{legal}. 

Let us first look at the structure of the Hilbert space in the light of our Hamiltonian.
According to the number of active sites in the clock register, 
the Hilbert space is a direct sum $\cH_{0}\oplus \cH_{1} \oplus \cH_{2} \oplus \cH_{3} \oplus \cdots$. The orthogonal subspaces $\cH_{j}$ are not coupled by our Hamiltonian.
We divide them into three cathegories.

First, we would like to avoid the subspaces $\cH_{j>1}$ with more than one active clock site. We can deal with them by a 2-local, $O(L)$ norm, clock-checking Hamiltonian 
%$H_{clock}$. For a pulse clock with a single active site $\ket{\mathtt{1}}$, the Hamiltonian is
\begin{eqnarray}
	H_{clock} = \sum_{\langle j,k\rangle} \ket{\mathtt{11}}\bra{\mathtt{11}}_{c_j,c_k}.
	\label{clockpenalty}
\end{eqnarray}
Its terms add an energy penalty for two neighboring (also within individual railroad switches) active sites in the clock. Although states such as 
$\ket{\varphi}_w\otimes\ket{\mathtt{010\cdots010}}_c$ are not directly violating \eqref{clockpenalty}, eigenstates of $H$ containing such clock registers have nonzero energy. An eigenstate with two trains in the clock register misses required terms coming from the propagation Hamiltonian. If they were all present, terms like $\ket{\varphi'}_w\otimes\ket{\mathtt{0\cdots0110\cdots0}}_c$
with an illegal clock configuration detected by \eqref{clockpenalty} would appear.
This line of reasoning goes along the lines of the clairvoyance lemma of Aharonov et. al. in \cite{GottesmanLine:07}.
It can be used to show that after adding $H_{clock}$, the energy of the states in the sectors $\cH_{j>1}$ is at least $O(L^{-1})$.

Second, we have to deal with $\cH_{0}$ -- the computationally ``dead'' subspace with no active site in the clock register. This is a problem of the pulse clock (see Section \ref{clocksection}). Noise which could kill the active site in the clock could thus also kill the computation. One way to prevent this is to energetically favor a single active site by adding a carefully balanced Hamiltonian of the form
\begin{eqnarray}
	H_{active} = - a_1 \sum_{j} \ket{\mathtt{1}}\bra{\mathtt{1}}_{c_j} + a_2 
	\sum_{\langle j,k \rangle} \ket{\mathtt{11}}\bra{\mathtt{11}}_{c_j,c_k},
	\label{balance}
\end{eqnarray}
where the sum in the second term is taken over neighboring qubits in the clock register.
Together with the clairvoyance lemma in action, this Hamiltonian favors a single active site over none (or over too many of them).
However, adding $H_{active}$ means that our total Hamiltonian can no longer be rewritten as a sum of 3-local projector terms with a common zero-energy ground state (and thus be frustration-free) \cite{footnoteQMA}.

Finally, we have the subspace $\cH_{1}$ with a single train, which can be further decomposed into subspaces corresponding to different initial states $\ket{\varphi_0}_w$ of the $n$ work qubits:
\begin{eqnarray}
	\cH_{1} = \bigoplus_{\varphi_0=0}^{2^n-1} \cH_{1}^{(\varphi_0)}.
\end{eqnarray}
Out of these, $\cH_{1}^{(0)} = \cH_{legal}$ is our computational subspace, corresponding to all states that can be reached by Schr\"odinger time evolution from the initial state \eqref{psi0}.
The eigenvectors of $H$ in the subspace $\cH_1$ are the momentum states $\ket{p,\varphi_0}$, labeled by momentum $p$ and the initial state of the work qubits $\ket{\varphi_0}_w$. For the badly initialized states with $\varphi_0\neq 0$, Lloyd \cite{AQC:Lloyd} showed that $\bra{p,\varphi_0\neq 0}H_{datainit}\ket{p,\varphi_0\neq 0} \geq O(L^{-1})$. Therefore, the energy gap between the computational subspace $\cH_{1}^{(0)}$ and the subspaces $\cH_{1}^{(\varphi_0\neq 0)}$ is again on the order of $O(L^{-1})$.

%%%%%%%%%%%%%%%%%%%%%%%%%%%%%%%%%%%%%%%%%%%%%%%%%%%%%%%%%%%%%%%%%%%%%%%%
%%%%%%%%%%%%%%%%%%%%%%%%%%%%%%%%%%%%%%%%%%%%%%%%%%%%%%%%%%%%%%%%%%%%%%%%

\section{The Second Model: A Universal 2-local Qubit-Qutrit Hamiltonian}
\label{section23}

Could we modify the Quantum-3-SAT Hamiltonian of Section \ref{switchsection} to make it only 2-local? This is probably too much to ask when restricting ourselves to qubits. Nevertheless, by using qutrits in place of some of the qubits, we can obtain a 2-local construction. First, we modify the model from Section \ref{switchsection}, and show that
the result is a quantum walk on a necklace -- a cycle with regularly spaced extra nodes connected to it. Using the results of \cite{necklaces}, we bound the required running time of this model.

We modify the gadget from Figure \ref{figureswitch} by using three qutrits and two qubits on each track, as in Figure \ref{figure23}. A qubit can be active $\ket{\mathtt{1}}$ or inactive $\ket{\mathtt{0}}$. On the other hand, each qutrit will have two active states $\ket{\mathtt{L}}$, $\ket{\mathtt{R}}$ (left and right) and one inactive state $\ket{\mathtt{0}}$. 
We label the possible positions of the active site in the clock register (the train) on the upper track 
by $u_1^{\mathtt{L}}, u_1^{\mathtt{R}}, u_2, u_3^{\mathtt{L}}, u_3^{\mathtt{R}}, u_4, u_5^{\mathtt{L}}$, $u_5^{\mathtt{R}}$ and denote the corresponding states of the clock register $\kets{u_1^{\mathtt{L}}}_c$,
$\kets{u_1^{\mathtt{R}}}_c$, etc. Similarly, we label the possible train positions on the bottom track $d_1^{\mathtt{L}},\dots,d_5^{\mathtt{R}}$
and the corresponding clock register states $\kets{l_1^{\mathtt{L}}}_c$, $\kets{l_1^{\mathtt{R}}}_c$, etc.

There are two types of 2-local transition rules for the upper track. First, we have 6 straightforward rules corresponding to the thick black lines on the upper track in Figure \ref{figure23}:
\begin{eqnarray}
	\begin{array}{ll}
		\kets{u_1^{\mathtt{L}}}\bras{c_{t-1}} + \kets{c_{t-1}}\bras{u_1^{\mathtt{L}}}, &
	\kets{u_2}\bras{u_1^{\mathtt{R}}} + \kets{u_1^{\mathtt{R}}}\bras{u_2},  \\
	\kets{u_3^{\mathtt{L}}}\bras{u_2} + \kets{u_2}\bras{u_3^{\mathtt{L}}},  &
	\kets{u_4}\bras{u_3^{\mathtt{R}}} + \kets{u_3^{\mathtt{R}}}\bras{u_4},  \\
	\kets{u_5^{\mathtt{L}}}\bras{u_4} + \kets{u_4}\bras{u_5^{\mathtt{L}}},  &
	\kets{c_t}\bras{u_5^{\mathtt{R}}} + \kets{u_5^{\mathtt{R}}}\bras{c_t}.
	\end{array}
	\label{uppertrack23}
\end{eqnarray}
%where $\kets{u_j^{\mathtt{L}}}$ and $\kets{u_j^{\mathtt{R}}}$ denote the two active states of a qutrit (the third state corresponds to the situation when the active site is elsewhere in the clock register).
Second, we have the transitions denoted by red half-circles in Figure \ref{figure23}. They are also 2-local,
involving a work qubit and a clock qutrit: 
\begin{eqnarray}
	\kets{1}\bras{1}_{q_1} \otimes \big( \kets{u_1^{\mathtt{R}}}\bras{u_1^{\mathtt{L}}} + \kets{u_1^{\mathtt{L}}}\bras{u_1^{\mathtt{R}}} \big), &&\label{t1x23}\\ 
	\sigma_{q_2}^{x} \otimes \big( \kets{u_3^{\mathtt{R}}}\bras{u_3^{\mathtt{L}}} + \kets{u_3^{\mathtt{L}}}\bras{u_3^{\mathtt{R}}} \big), &&\label{t3x23}\\ 
	\kets{1}\bras{1}_{q_1} \otimes \big( \kets{u_5^{\mathtt{R}}}\bras{u_5^{\mathtt{L}}} + \kets{u_5^{\mathtt{L}}}\bras{u_5^{\mathtt{R}}} \big).&&
	\label{t5x23}
\end{eqnarray}
The rules \eqref{t1x23} and \eqref{t5x23} are conditioned on the state of the control qubit $q_1$, and \eqref{t3x23} flips the target work qubit $q_2$ when transitioning between $\kets{u_3^{\mathtt{L}}}$ and $\kets{u_3^{\mathtt{R}}}$
in the clock register.

Similar transition rules apply to the bottom track. The straightforward terms are
\begin{eqnarray}
	\begin{array}{ll}
		\kets{l_1^{\mathtt{L}}}\bras{c_{t-1}} + \kets{c_{t-1}}\bras{l_1^{\mathtt{L}}}, &
	\kets{l_2}\bras{l_1^{\mathtt{R}}} + \kets{l_1^{\mathtt{R}}}\bras{l_2},  \\
	\kets{l_3^{\mathtt{L}}}\bras{l_2} + \kets{l_2}\bras{l_3^{\mathtt{L}}},  &
	\kets{l_4}\bras{l_3^{\mathtt{R}}} + \kets{l_3^{\mathtt{R}}}\bras{l_4},  \\
	\kets{l_5^{\mathtt{L}}}\bras{l_4} + \kets{l_4}\bras{l_5^{\mathtt{L}}},  &
	\kets{c_t}\bras{l_5^{\mathtt{R}}} + \kets{l_5^{\mathtt{R}}}\bras{c_t}.
	\end{array}
\end{eqnarray}
and the analogues of
\eqref{t1x23} and \eqref{t5x23} for the lower track are conditioned on the control qubit $q_1$ being $\ket{0}$, i.e.
\begin{eqnarray}
	\kets{0}\bras{0}_{q_1} \otimes \big( \kets{l_1^{\mathtt{R}}}\bras{l_1^{\mathtt{L}}} + \kets{l_1^{\mathtt{L}}}\bras{l_1^{\mathtt{R}}} \big), \\	 
	\kets{0}\bras{0}_{q_1} \otimes \big( \kets{l_5^{\mathtt{R}}}\bras{l_5^{\mathtt{L}}} + \kets{l_5^{\mathtt{L}}}\bras{l_5^{\mathtt{R}}} \big).\,
\end{eqnarray}
Finally, the analogue of \eqref{t3x23} is simply 
\begin{eqnarray}
	\kets{u_3^{\mathtt{R}}}\bras{u_3^{\mathtt{L}}} + \kets{u_3^{\mathtt{L}}}\bras{u_3^{\mathtt{R}}},
	\label{untouched23}
\end{eqnarray}
leaving the work qubit $q_2$ untouched.
Summing all the terms in \eqref{uppertrack23}-\eqref{untouched23} gives us the 2-local qubit-qutrit railroad switch Hamiltonian $H_{23}^{(t)}$.

\begin{figure}
	\begin{center}
	\includegraphics[width=3in]{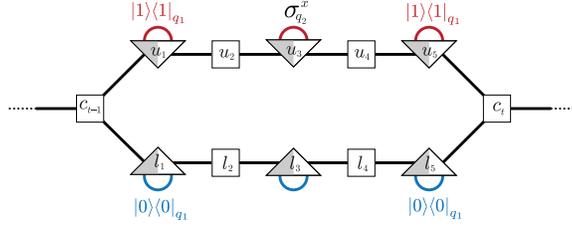} 
	\end{center}
	\caption{The 2-local {\em railroad switch} gadget made from qubits (squares) and qutrits (triangles). 
	Simple transitions of the train are denoted by black lines. Transitions between two active states of a qutrit are denoted by half-circles.
	\label{figure23}}
\end{figure}

\begin{figure}
	\begin{center}
	\includegraphics[width=2in]{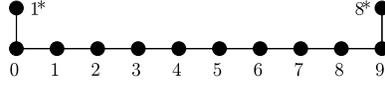} 
	\end{center}
	\caption{The geometry of connections in $H_{23}^{(t)}$ restricted to the subspace $\cH_{g23}^{(t)}$ (within one gadget).
	\label{figurecomb}}
\end{figure}

\begin{figure}
	\begin{center}
	\includegraphics[width=1.6in]{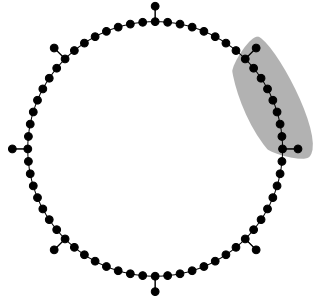} 
	\end{center}
	\caption{The necklace graph depicting the connections in $H_{23}$ in a particular basis. The shaded region corresponds to states within one railroad-switch gadget.
	\label{figurenecklace}}
\end{figure}

Similarly to what we did in Section \ref{dynamicssection},
let us examine where the transitions in this Hamiltonian take the state
\begin{eqnarray}
	\ket{\psi_{t-1}} 
		&=& \big(
		a \underbrace{\ket{0}_{q_1} \ket{\alpha}_{q_2,\dots}}_{\ket{\varphi_{0\dots}}} 
		+ b \underbrace{\ket{1}_{q_1} \ket{\beta}_{q_2,\dots}}_{\ket{\varphi_{1\dots}}}
		\big)
		\otimes \ket{c_{t-1}}_c 
\end{eqnarray}
where $a \ket{\varphi_{0\dots}} + b \ket{\varphi_{1\dots}}$ is a general state of the work register
and the active site of clock register is at the qubit $c_{t-1}$.
Within the gadget, the transitions in $H_{23}$ move this state around a 12-dimensional subspace 
$\cH_{g32}^{(t)}$ with the geometry of allowed transitions depicted in Figure \ref{figurecomb}. The basis of this subspace contains first the endpoint states
\begin{eqnarray}
	\kets{\psi_{t-1}^{(0)}} &=& 
	\big( a \ket{\varphi_{0\dots}} + b \ket{\varphi_{1\dots}}
		\big)
		\otimes \ket{c_{t-1}}_c  \\
	\kets{\psi_{t}^{(9)}} 
	&=& \left(	a  \ket{\varphi_{0\dots}} +	b \ket{\varphi'_{1\dots}} \right) 
			\otimes \ket{c_{t}}_c,
\end{eqnarray}
where $\kets{\varphi'_{1\dots}} = \sigma^x_{q_2} \kets{\varphi_{1\dots}}$ has the target qubit $q_2$ flipped. Second, we have the sequence  
\begin{eqnarray}
	\kets{\psi_{t-1}^{(1)}} 
	&=& a \ket{\varphi_{0\dots}} \otimes \ket{l_1^{\mathtt{L}}}_c
		+	b \ket{\varphi_{1\dots}} \otimes \ket{u_1^{\mathtt{L}}}_c, \\
	\kets{\psi_{t-1}^{(2)}} 
	&=& a \ket{\varphi_{0\dots}} \otimes \ket{l_1^{\mathtt{R}}}_c
		+	b \ket{\varphi_{1\dots}} \otimes \ket{u_1^{\mathtt{R}}}_c, \\
	\kets{\psi_{t-1}^{(3)}} 
	&=& a \ket{\varphi_{0\dots}} \otimes \ket{l_2^{\phantom{\mathtt{R}}}}_c
		+	b \ket{\varphi_{1\dots}} \otimes \ket{u_2^{\phantom{\mathtt{R}}}}_c, \\
	\kets{\psi_{t-1}^{(4)} }
	&=& a \ket{\varphi_{0\dots}} \otimes \ket{l_3^{\mathtt{L}}}_c
		+	b \ket{\varphi_{1\dots}} \otimes \ket{u_3^{\mathtt{L}}}_c, \\
	\kets{\psi_{t-1}^{(5)}} 
	&=& a \ket{\varphi_{0\dots}} \otimes \ket{l_3^{\mathtt{R}}}_c
		+	b \ket{\varphi'_{1\dots}} \otimes \ket{u_3^{\mathtt{R}}}_c, \\
	\kets{\psi_{t-1}^{(6)}} 
	&=& a \ket{\varphi_{0\dots}} \otimes \ket{l_4^{\phantom{\mathtt{R}}}}_c
		+	b \ket{\varphi'_{1\dots}} \otimes \ket{u_4^{\phantom{\mathtt{R}}}}_c, \\
	\kets{\psi_{t-1}^{(7)}} 
	&=& a \ket{\varphi_{0\dots}} \otimes \ket{l_5^{\mathtt{L}}}_c
		+	b \ket{\varphi'_{1\dots}} \otimes \ket{u_5^{\mathtt{L}}}_c, \\
	\kets{\psi_{t-1}^{(8)}}
	&=& a \ket{\varphi_{0\dots}} \otimes \ket{l_5^{\mathtt{R}}}_c
		+	b \ket{\varphi'_{1\dots}} \otimes \ket{u_5^{\mathtt{R}}}_c.
\end{eqnarray}
Finally, the last two states in the basis are the two `blind-alley' states, in which the active site tries to use the track which it is not supposed to take
\begin{eqnarray}
	\kets{\psi_{t-1}^{(1*)}}_g 
	&=& a \ket{\varphi_{0\dots}} \otimes \ket{u_1^{\mathtt{L}}}_c
		+	b \ket{\varphi_{1\dots}} \otimes \ket{l_1^{\mathtt{L}}}_c, \\
	\kets{\psi_{t-1}^{(8*)}}_g 
	&=& a \ket{\varphi_{0\dots}} \otimes \ket{u_8^{\mathtt{R}}}_c
		+	b \ket{\varphi'_{1\dots}} \otimes \ket{l_8^{\mathtt{R}}}_c.
\end{eqnarray}
Omitting the connection to the region outside the gadget, this 12-dimensional subspace $\cH_{g23}^{(t)}$ is invariant under the transitions in $H_{23}^{(t)}$. 
It is convenient to analyze the dynamics of $H_{23}^{(t)}$ using the basis just given.
In this basis, the Hamiltonian within one gadget has a very simple form -- the adjacency matrix of a line $\kets{\psi_{t-1}^{(0)}},\kets{\psi_{t-1}^{(1)}},\dots,\kets{\psi_{t-1}^{(8)}},\kets{\psi_t^{(9)}}$ with two additional nodes $\kets{\psi_{t-1}^{(1*)}}$ and $\kets{\psi_{t-1}^{(8*)}}$ as in Figure \ref{figurecomb}. Moreover, 
we can give our clock register a very regular form by arranging 8 clock qubits between two successive gadgets. 
We can then choose to look at the dynamics of our system in a basis consisting of states of the form \eqref{psit} (for the non-gadget parts) and of the bases just given for individual gadgets.
Similarly to what we found in Section \ref{dynamicssection}, in this basis 
the Hamiltonian of our model has a very simple form. It is an adjacency matrix of a ring with extra lines sticking out of it on every 9-th point as in Figure \ref{figurenecklace}. 
This allows us to forget about the content of the work register when solving the dynamics of the time evolution with our Hamiltonian.

Finally, we need to find the required running time for the qubit-qutrit model. The dynamics of the time evolution generated by the Hamiltonian is a quantum walk on a necklace from Figure \ref{figurenecklace}. In an upcoming paper by Nagaj and Reitzner \cite{necklaces}, we investigate the mixing properties of this type of quantum walks. 
Again, the mixing is fast. When we start in the initial state \eqref{circleclockinit}, and let the system evolve with $H_{23}$ for a random time $0\leq \tau \leq T$ with $T = O(L^2)$ (up to logarithmic factors) and measure the output work qubit, our model is universal for $\BQP$.

%%%%%%%%%%%%%%%%%%%%%%%%%%%%%%%%%%%%%%%%%%%%%%%%%%%%%%%%%
%%%%%%%%%%%%%%%%%%%%%%%%%%%%%%%%%%%%%%%%%%%%%%%%%%%%%%%%%

\section*{Acknowledgments}
The author would like to thank Andrew Landahl and Avi Hassidim for the discussions at the QIP 2008 poster session, Seth Lloyd for helpful discussions and access to his unpublished manuscript, Stephen Jordan for useful comments on improving the manuscript
and finally, Diego DeFalco for pointing out the notion of the {\em switch} in Feynman's paper and further references. Part of this work was done while the author was finishing his Ph.D. thesis at MIT. The author gratefully acknowledges support from the Slovak Research and Development Agency under the contract No. APVV-0673-07 and contract No. LPP-0430-09 and from the W. M. Keck Foundation Center for Extreme Quantum Information Theory at MIT.

%%%%%%%%%%%%%%%%%%%%%%%%%%%%%%%%%%%%%%%%%%%%%%%%%%%%%%%%%
%%%%%%%%%%%%%%%%%%%%%%%%%%%%%%%%%%%%%%%%%%%%%%%%%%%%%%%%%

\appendix

\section{Continuous-time Quantum Walks in 1D}
\label{walksection}

First, in Section \ref{plainline} we analyze the continuous-time quantum walk on a line and prove a lemmas about the mixing of this walk. Second, 
%in Section \ref{endsline} look at the walk with additional boundary terms and 
in Section \ref{circleline} we analyze the quantum walk on a cycle and prove another mixing lemma used in the proof of universality of the {\em railroad switch} Hamiltonian Computer in Section \ref{switchsection}.

%%%%%%%%%%%%%%%%%%%%%%%%%%%%%%%%%
%%%%%%%%%%%%%%%%%%%%%%%%%%%%%%%%%
\subsection{Quantum Walk on a Line}
\label{plainline}
Consider a continuous time quantum walk on a line of length $L$, where the Hamiltonian is the negative of the adjacency matrix for the line 
\begin{eqnarray}
	H_1 = - \sum_{j=1}^{L-1} 
	\left(\ket{j}\bra{j+1} + \ket{j+1}\bra{j}\right).
	\label{H1hamiltonian}
\end{eqnarray}
The eigenvalues of this Hamiltonian are
\begin{eqnarray}
	\lambda_j = - 2 \cos \left(\frac{j\pi}{L+1}\right),
	\label{eigenval}
\end{eqnarray}
for $j=1\dots L$, while the corresponding eigenvectors $\kets{\phi^{(j)}} = \sum_{k=1}^{L} \phi^{(j)}_k \ket{k}$ have components
\begin{eqnarray}
	\phi^{(j)}_k = \sqrt{\frac{2}{L+1}} \, 
		\sin \left( \frac{j k \pi}{L+1} \right).
	\label{eigenvec}
\end{eqnarray}
Consider the time evolution of a particular basis state $\ket{c}$. The probability of finding the system in the 
%%% XXX
basis 
%%% XXX
state $\ket{m}$ at some time $\tau$ can be found by expanding $\ket{c}$ and $\ket{m}$ in the basis of the eigenvectors \eqref{eigenvec}:
\begin{eqnarray}
	p_{\tau}(m|c) = \left| 
				\bra{m} e^{-iH\tau} \ket{c}
			\right|^2
		= \sum_{j,k=1}^{L} e^{-i (\lambda_j -\lambda_k) \tau} 
			\phi^{(j)}_m
			\phi^{(j)*}_c
			\phi^{(k)*}_m
			\phi^{(k)}_c.
\end{eqnarray}
Because the time evolution (according to the Schr\"odinger equation) is unitary, this probability $p_{\tau}(m|c)$ does not converge. On the other hand, let us define the time average of $p_{\tau}(m|c)$
for time $0\leq\tau\leq\tau_{20}$ as
\begin{eqnarray}
	\bar{p}_{\tau_{20}} (m|c) 
		= \frac{1}{\tau_{20}} \int_0^{\tau_{20}} p_{\tau} (m|c) d\tau.
	\label{Paverage}
\end{eqnarray}
As we will show below in Lemma \ref{convergelemma}, this average probability distribution does converge to a limiting distribution $\pi(m|c)$, defined
as the $\tau_{20}\rightarrow\infty$ limit of the average probability distribution \eqref{Paverage}. All the eigenvalues \eqref{eigenval} are different, so we can express the limiting distribution as
\begin{eqnarray}
	\pi(m|c) 
	= \lim_{\tau_{20}\rightarrow\infty} \bar{p}_{\tau_{20}}(m|c) 
	= \sum_{j=1}^{L} \big|\phi^{(j)}_m\big|^2 
		\big|\phi^{(j)}_c\big|^2,
	\label{deflimdistrib}
\end{eqnarray}
which in this case is
\begin{eqnarray}
	\pi(m|c) 
	 = \frac{2+\delta_{m,c} + \delta_{m,L+1-c}}{2(L+1)}.
	 \label{Plimiting}
\end{eqnarray}

According to the following lemma, the average probability distribution \eqref{Paverage} converges to the limiting distribution $\pi(m|c)$.
\begin{lemma}
\label{convergelemma}
Consider a continuous time quantum walk on a line of length $L$, where the Hamiltonian is the negative of the adjacency matrix for the line. Let the system evolve for time $\tau \leq \tau_{20}$ chosen uniformly at random, starting in a position basis state $\ket{c}$.
The average probability distribution $\bar{p}_{\tau_{20}}(\cdot|c)$ converges to the limiting probability distribution $\pi(\cdot|c)$ as
\begin{eqnarray}
	\sum_{m=1}^{L} \left| \bar{p}_{\tau_{20}}(m|c) - \pi(m|c) \right|
	\leq 
		O\left(\frac{L}{\tau_{20}}\right).
	\label{limitinglemma}
\end{eqnarray}
\end{lemma}

\begin{proof}
First, recall Lemma 4.3 of \cite{CA:AAKVwalk:01} for the total variation distance of the probability distribution $\bar{p}_{\tau_{20}}$ from the limiting distribution, saying
\begin{eqnarray}
	\sum_m \left| \bar{p}_{\tau_{20}}(m|c) - \pi(m|c) \right| \leq 
		\frac{2}{\tau_{20}} \sum_{\lambda_j\neq \lambda_k} 
		\frac{\big|\phi^{(j)}_c\big|^2}{ |\lambda_j-\lambda_k|}.
	\label{convergebound}
\end{eqnarray}
Using \eqref{eigenval} and \eqref{eigenvec}, we can bound the expression on the right of \eqref{convergebound}. When $j$ is close to $k$, i.e. $|j-k|\leq C_1$, we can write
\begin{eqnarray}
	\frac{\big|\phi^{(j)}_c\big|^2}{|\lambda_j-\lambda_k|} < 2.
\end{eqnarray}
On the other hand, for $|j-k|>C_1$ we can bound
\begin{eqnarray}
	\frac{\big|\phi^{(j)}_c\big|^2}{|\lambda_j-\lambda_k|} < \frac{C_2}{L+1},
\end{eqnarray}
with $C_1$ and $C_2$ constants independent of $L$. 
Inserting into \eqref{convergebound}, we have
\begin{eqnarray}
	\sum_{m=1}^{L} \left| \bar{p}_{\tau_{20}}(m|c) - \pi(m|c) \right|
	\leq 
		\frac{8 C_1 L}{\tau_{20}} + \frac{C_2 L}{\tau_{20}}
		= O\left(\frac{L}{\tau_{20}}\right).
\end{eqnarray}
\end{proof}

%%%%%%%%%%%%%%%%%%%%%%%%%%%%%%%%%%%%%%%%%%%%%%%%%%%%%%%%%%
%%%%%%%%%%%%%%%%%%%%%%%%%%%%%%%%%%%%%%%%%%%%%%%%%%%%%%%%%%

\subsection{Quantum Walk on a Circle}
\label{circleline}

If the geometry of the system is a closed loop of length $L$ instead of a line, the Hamiltonian \eqref{H1hamiltonian} gets an additional wrap-around term.
\begin{eqnarray}
	H_{loop} = 
	-\left(\ket{L}\bra{1} + \ket{1}\bra{L}\right)
	- \sum_{j=1}^{L-1} 
	\left(\ket{j}\bra{j+1} + \ket{j+1}\bra{j}\right).
	\label{HLoophamiltonian}
\end{eqnarray}
The eigenvalues of this Hamiltonian are
\begin{eqnarray}
	\lambda_j = - 2 \cos p_j,
	\label{cycleeigenval}
\end{eqnarray}
corresponding to plain waves with momenta
\begin{eqnarray}
	p_j = \frac{2\pi j}{L},
	\label{cyclemomenta}
\end{eqnarray}
for $j=0\dots L-1$. The corresponding eigenvectors $\kets{\phi^{(j)}} = \sum_{k=1}^{L} \phi^{(j)}_k \ket{k}$ have components
\begin{eqnarray}
	\phi^{(0)}_k &=& \frac{1}{\sqrt{L}}, \\
	\phi^{(j)}_k &=& \sqrt{\frac{2}{L}} e^{ip_j k}, \qquad j=1,\dots,L-1,
	\label{cycleeigenvec}
\end{eqnarray}
These can be combined to make real eigenvectors. For our analysis, it will be enough to consider a line with even length, and only the cosine plain waves:
\begin{eqnarray}
	\phi^{(0)}_k &=& \frac{1}{\sqrt{L}}, \\
	\phi^{(j)}_k &=& \sqrt{\frac{2}{L}} \cos(p_j k), \qquad j=1,\dots,\frac{L}{2}.
	\label{cycleeveneigenvec}
\end{eqnarray}

The limiting distribution on the cycle when starting from site $c$ is 
\begin{eqnarray}
	\pi(m|c) 
	 = \frac{1}{L}- \frac{2}{L^2}.
	 \label{Pcyclelimiting2}
\end{eqnarray}
for all points $m$ except for $m=c$ (return back) and $m=L+1-c$
(the point across the cycle), where we have
\begin{eqnarray}
	\pi(m|c) 
	 = \frac{2}{L} - \frac{2}{L^2}.
	 \label{Pcyclelimiting1}
\end{eqnarray}
As in Section \ref{plainline}, we will again utilize Lemma 4.3 of \cite{CA:AAKVwalk:01} to prove the convergence of the time-averaged probability distribution to this limiting distribution. On the right side of \eqref{convergebound}, we now have
\begin{eqnarray}
	\big|\phi_0^{(j)}\big|^2 \leq \frac{2}{L}.
\end{eqnarray}
The sum over the non-equal eigenvalues
\begin{eqnarray}
	S_{\circ} = \sum_{\lambda_j\neq \lambda_k} 
		\frac{1}{ |\lambda_j-\lambda_k|}
		\label{cycleeigsum}
\end{eqnarray}
is now more complicated, because of the degeneracy of the spectrum.
\begin{figure}
	\begin{center}
	\includegraphics[width=2in]{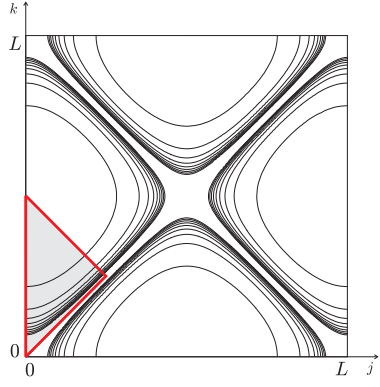}
	\end{center}
	\caption{A contour plot of $\left|\lambda_j - \lambda_k\right|^{-1}$ with $\lambda_j$ from \eqref{cycleeigenval}. The sum in \eqref{cycle16} is over the marked region.}
	\label{appc:figcircleeig}
\end{figure}
We plot $|\lambda_j-\lambda_k|^{-1}$ in Figure \ref{appc:figcircleeig}.
Because of the symmetrical way $\lambda_j$ \eqref{cycleeigenval} arise, assuming $L$ is divisible by 4, we can express \eqref{cycleeigsum} as
\begin{eqnarray}
	S_{\circ} 
	&=& 
	\sum_{\lambda_j\neq\lambda_k} 
		\frac{1}{ |\lambda_j-\lambda_k|} \\
% 	&=& \sum_{j=0}^{L-1} \sum_{k\neq j, L-j} 		\frac{1}{ |\lambda_j-\lambda_k|}	\\
%	&=& 	8 \sum_{j=0}^{L/2} \sum_{k=j+1}^{L/2} 		\frac{1}{ |\lambda_j-\lambda_k|} 	\\
	&\leq& 
	16 \sum_{j=0}^{L/4} \sum_{k=j+1}^{L/2-j}
		\frac{1}{ |\lambda_j-\lambda_k|}
		\label{cycle16}
	\\
	&=& 
	16 \underbrace{\sum_{k=1}^{L/2}
		\frac{1}{ |\lambda_0-\lambda_k|}}_{A_0}
		+
	16 \underbrace{\sum_{j=1}^{L/4} \sum_{k=j+1}^{L/2-j}
		\frac{1}{ |\lambda_j-\lambda_k|}}_{A_1}.
		\label{cyclesumbounds}
\end{eqnarray}
The term $A_0$ comes from $j=0$. We can bound the sum by an integral, taking $x=\frac{\delta k}{L}$, obtaining
\begin{eqnarray}
	A_0 = \frac{1}{2}\sum_{k=1}^{L/2}
		\frac{1}{ 1-\cos\left(2\pi \frac{k}{L}\right)} 
	\leq \frac{L}{2} \int_{\frac{1}{L}}^{\frac{1}{2}}
		\frac{dx}{ 1-\cos(2\pi x)}.
\end{eqnarray}
For $x \in \left[0,\half\right]$, we can bound
\begin{eqnarray}
	1-\cos{2 \pi x} \geq 2 \pi x^2,
\end{eqnarray}
resulting in
\begin{eqnarray}
	A_0 \leq \frac{L}{4\pi} \int_{\frac{1}{L}}^{\frac{1}{2}}
		\frac{dx}{x^2}
		= \frac{1}{4\pi} \left(L^2 - 2L\right) = O(L^2).
		\label{a0bound}
\end{eqnarray}
We bound the other term, $A_1$, in \eqref{cyclesumbounds} by an integral as well:
\begin{eqnarray} 
	A_1 &=& \frac{1}{2}\sum_{j=1}^{L/4} \sum_{k=j+1}^{L/2-j}
		\frac{1}{ \cos\left(2\pi\frac{k}{L}\right)
			- \cos\left(2\pi\frac{j}{L}\right)} \\
		&\leq&
		\frac{L^2}{2} \int_{\frac{1}{L}}^{\frac{1}{4}-\frac{1}{L}}
		dy 
		\int_{\frac{1}{L}+y}^{\frac{1}{2}-y}
		\frac{dx}{\cos\left(2\pi x\right)
			- \cos\left(2\pi y\right)}.		
\end{eqnarray} 
Again, we can lower bound the difference in eigenvalues for $y\in \left[0,\frac{1}{4}\right]$ and $x\in \left[y,\frac{1}{2}-y\right]$ by
\begin{eqnarray}
	\cos\left(2\pi x\right)
			- \cos\left(2\pi y\right) \geq 2\pi (x^2-y^2),
\end{eqnarray}
allowing us to write
\begin{eqnarray} 
	A_1 &\leq& 
		\frac{L^2}{4\pi} \int_{\frac{1}{L}}^{\frac{1}{4}-\frac{1}{L}}
		dy 
		\int_{\frac{1}{L}+y}^{\frac{1}{2}-y}
		\frac{dx}{x^2-y^2} \\
	&=&
		\frac{L^2}{4\pi} \int_{\frac{1}{L}}^{\frac{1}{4}-\frac{1}{L}}
		dy 
		\left[
			\frac{\log\left(\frac{x-y}{x+y}\right)}{2y}
		\right]_{\frac{1}{L}+y}^{\frac{1}{2}-y} \\
	&=&
		\frac{L^2}{4\pi} \int_{\frac{1}{L}}^{\frac{1}{4}-\frac{1}{L}}
		\frac{dy}{2y} 
		\underbrace{\left[
			\log\left(1-4y\right)
			+ 
			\log\left(1+2yL\right)
		\right]}_{R}.
		\label{a1final}
\end{eqnarray} 
As $y$ in \eqref{a1final} is at most $\frac{1}{4}-\frac{1}{L}$,
we can bound $R$ by
\begin{eqnarray}
	|R| \leq \log L.
\end{eqnarray}
Finally, this results in
\begin{eqnarray}
	A_1 \leq \frac{L^2 \log L}{2\pi} \int_{\frac{1}{L}}^{\frac{1}{4}-\frac{1}{L}}
		\frac{dy}{2y} 
		= \frac{L^2 \log L}{2\pi} 
		\underbrace{\log \left(\frac{\frac{1}{4}-\frac{1}{L}}{\frac{1}{L}}\right)}_{\leq \log L}
		\leq O(L^2 \log^2 L).
		\label{a1bound}
\end{eqnarray}
Putting \eqref{a0bound} and \eqref{a1bound} into \eqref{cyclesumbounds}, we obtain
\begin{eqnarray}
	S_{\circ} \leq O(L^2 \log^2 L).
\end{eqnarray}
Lemma 4.3 of \cite{CA:AAKVwalk:01} (see \eqref{convergebound}) then reads
\begin{eqnarray}
	\sum_m \left| \bar{p}_{\tau_{20}}(m|c) - \pi(m|c) \right| \leq 
		\frac{2}{\tau_{\circ}} \frac{1}{L} O(L^2 \log^2 L)
		= O\left(\frac{L \log^2 L}{\tau_\circ}\right).
	\label{convergeboundcycle}
\end{eqnarray}
Thus, for $\tau_{circ} = \ep O(L \log^2 L)$, the time-averaged distribution becomes $\ep$-close to the limiting distribution. Using the bound on the total variation distance we just proved, it is straightforward to obtain the following lemma which we use in Section \ref{switchsection}:
\begin{lemma}
	\label{cyclelemma}
Consider a continuous time quantum walk on a cycle of length $L$ (divisible by 4), where the Hamiltonian is the negative of the adjacency matrix for the cycle. Let the system evolve for a time $\tau \leq \tau_{\circ}$ chosen uniformly at random, starting in a position basis state $\ket{0}$.
The probability to measure a position state $\ket{t}$ farther than $L/6$ from the starting point (the farther two thirds of the cycle) is then bounded from below as $p_{\circ}\geq\frac{2}{3}-\frac{1}{3L}-O\left(\frac{L \log^2 L}{\tau_{\circ}}\right)$. 
\end{lemma}
\begin{proof}
We proceed as in the proof of Lemma \ref{convergelemma} in Section \ref{plainline}. 
Let us call the farther two thirds of the cycle (see Figure \ref{figurecircle} in Section \ref{switchsection}) the success region (SR). When we choose the time $\tau\leq \tau_{\circ}$ uniformly at random, the probability to measure a state $\ket{t}$ with $t\in SR$ is\begin{eqnarray}
	p_{\circ} = \sum_{m\in SR} 
			 \bar{p}_{\tau_{\circ}}(m|c),
\end{eqnarray}
Using the bound on the total variation distance \eqref{convergeboundcycle} we just proved and the formulae for the limiting distribution \eqref{Pcyclelimiting2},\eqref{Pcyclelimiting1}, we have
\begin{eqnarray}
	O\left(\frac{L \log^2 L}{\tau_{\circ}}\right) 
	&\geq& 
		\sum_{m=1}^{L} 
			\left| \bar{p}_{\tau_{\circ}}(m|c) - \pi(m|c) \right| \\
	&\geq& 
		\sum_{m\in SR} 
			\left| \bar{p}_{\tau_{\circ}}(m|c) - \pi(m|c) \right| \\
	&\geq& 
		\left| \sum_{m\in SR} 
			 \bar{p}_{\tau_{\circ}}(m|c) 
		-
		\sum_{m\in SR} 
				\pi(m|c) \right| \\
	&=& 
		\left| p_{\circ} - \frac{2}{3} + \frac{1}{3L} + O\left(\frac{1}{L}\right) 
			\right|.
\end{eqnarray} 
Therefore, the probability of finding the chain in state $\ket{\psi_{t\in SR}}$ at a random time $\tau \leq \tau_{\circ}$ is bounded from below by
\begin{eqnarray}
p_{\circ} \geq \frac{2}{3} - \frac{1}{3L} - O\left(\frac{L \log^2 L}{\tau_{\circ}}\right).
\end{eqnarray}
\end{proof}
It is thus enough to wait a random time not larger than $O(L\log^2 L)$ to find the state of the system in the success region with probability close to two thirds.

%%%%%%%%%%%%%%%%%%%%%%%%%%%%%%%%%%%%%%%%%%%%%%%%%%%%%%%%%

%%%%%%%%%%%%%%%%%%%%%%%%%%%%%%%%%%%%%%%%%%%%%%%%%%%%%%%%%%%%%
%%%%%%%%%%%%%%%%%%%%%%%%%%%%%%%%%%%%%%%%%%%%%%%%%%%%%%%%%%%%%

%%%%%%%%%%%%%%%%%%%%%%%%%%%%%%%%%%%%%%%%%%%%%%%%
%%%%%%%%%%%%%%%%%%%%%%%%%%%%%%%%%%%%%%%%%%%%%%%%


\begin{thebibliography}{99}
\bibitem{footnoteLH}
A promise problem: Is the ground state energy of a 3-local Hamiltonian less than some $a$, or more than some $b$, with $(b-a)>1/poly(L)$.
\bibitem{footnoteQMA}
However, the term \eqref{balance} can be used to find a new proof of $\QMA$ universality for the local Hamiltonian problem.
\end{thebibliography}
\end{document}